\newif\ifdouble

\doubletrue
\ifdouble
  \documentclass[journal]{IEEEtran}
\else
  \documentclass[journal,draftclsnofoot,onecolumn]{IEEEtran}
\fi

\usepackage[utf8]{inputenc}
\usepackage[hidelinks]{hyperref}
\usepackage{amsmath} 
\usepackage{amssymb}
\usepackage{amsthm}
\usepackage{graphicx}
\usepackage{color}
\usepackage{enumitem}
\usepackage[algo2e,vlined,linesnumbered]{algorithm2e}

\newcommand{\R}{\mathbb{R}}
\newcommand{\N}{\mathbb{N}}

\newtheorem{assum}{Assumption}
\newtheorem{definition}[assum]{Definition}

\newtheorem{theorem}[assum]{Theorem}
\newtheorem{remark}[assum]{Remark}
\newtheorem{lemma}[assum]{Lemma}
\newtheorem{coro}[assum]{Corollary}
\newtheorem{problem}[assum]{Problem}

\graphicspath{{Images/}}

\begin{document}
\title{Hierarchical decomposition of LTL synthesis problem for nonlinear control systems}

\author{Pierre-Jean Meyer}

\author{Pierre-Jean Meyer, Dimos V.\ Dimarogonas% <-this % stops a space
\thanks{This work was supported by the H2020 ERC Starting Grant BUCOPHSYS, the EU H2020 AEROWORKS project, the EU H2020 Co4Robots project, the Swedish Foundation for Strategic Research, the Swedish Research Council and the KAW Foundation.}%
\thanks{P.-J.\ Meyer is with the Department of Electrical Engineering and Computer Sciences at University of California, Berkeley, CA 94720-1770, USA (email: pjmeyer@berkeley.edu).}%
\thanks{D.\ V.\ Dimarogonas is with the Department of Automatic Control at KTH Royal Institute of Technology, 10044 Stockholm, Sweden (email: dimos@kth.se).}% <-this % stops a space
% \thanks{Manuscript received April 19, 2005; revised August 26, 2015.}
}

\maketitle

%%%%%%%%%%%%%%%%%%%%%%%%%%%%%%%%%%%%%%%%%%%%%%%%%%%%%%%%%%%%%%%%%%%%%%%%%%%%%%%%%%%%%%%%%%%%%%%%%%%%%%%%%%%%
%%%%%%%%%%%%%%%%%%%%%%%%%%%%%%%%%%%%%%%%%%%%%%%%%%%%%%%%%%%%%%%%%%%%%%%%%%%%%%%%%%%%%%%%%%%%%%%%%%%%%%%%%%%%

\begin{abstract}
This paper deals with the control synthesis problem for a continuous nonlinear dynamical system under a Linear Temporal Logic (LTL) formula.
The proposed solution is a top-down hierarchical decomposition of the control problem involving three abstraction layers of the problem, iteratively solved from the coarsest to the finest.
The LTL planning is first solved on a small transition system only describing the regions of interest involved in the LTL formula.
For each pair of consecutive regions of interest in the resulting accepting path satisfying the LTL formula, a discrete plan is then constructed in the partitioned workspace to connect these two regions while avoiding unsafe regions.
Finally, an abstraction refinement approach is applied to synthesize a controller for the dynamical system to follow each discrete plan.
% The first two steps let us find a discrete plan satisfying the LTL formula with reduced computation need, even for large workspaces.
The second main contribution, used in the third abstraction layer, is a new monotonicity-based method to over-approximate the finite-time reachable set of any continuously differentiable system.
The proposed framework is demonstrated in simulation for a motion planning problem of a mobile robot modeled as a disturbed unicycle.
\end{abstract}


%%%%%%%%%%%%%%%%%%%%%%%%%%%%%%%%%%%%%%%%%%%%%%%%%%%%%%%%%%%%%%%%%%%%%%%%%%%%%%%%%%%%%%%%%%%%%%%%%%%%%%%%%%%%
%%%%%%%%%%%%%%%%%%%%%%%%%%%%%%%%%%%%%%%%%%%%%%%%%%%%%%%%%%%%%%%%%%%%%%%%%%%%%%%%%%%%%%%%%%%%%%%%%%%%%%%%%%%%
\section{Introduction}
\label{sec intro}
Control synthesis and planning for continuous dynamical systems under high-level specifications, such as Linear Temporal Logic (LTL) formulas~\cite{baier2008principles}, usually cannot be solved directly on the continuous dynamics.
The classical solutions to this problem thus rely on a two-step approach, where we first create a finite abstraction (the \emph{abstract} or \emph{symbolic model}) of the continuous dynamical system (the \emph{concrete model}), and then leverage formal methods from the field of computer science to synthesize a controller for the abstraction to satisfy the high-level specifications.
Provided that the abstraction was created to obtain some behavioral relationship (such as alternating simulation~\cite{tabuada2009symbolic} or feedback refinement relation~\cite{reissig2016}) between the concrete and abstract models, the controller obtained on the abstraction can then be concretized into a controller for the concrete model to satisfy the desired specifications.

This topic recently received significant interest resulting in various abstraction methods such as designing local feedback controllers between any two neighboring cells of a state space partition to guarantee the creation of a deterministic abstraction~\cite{gol2013time,boskos2015decentralized}, considering infinite-time reachability analysis of neighboring cells~\cite{nilsson2014incremental,yang2017fuel}, or fixed and finite-time reachability analysis~\cite{coogan2015mixed,reissig2016}, which we consider in this paper.
While the combined results of all these approaches cover a wide range of dynamical systems and control objectives, when taken separately most of these approaches (as well as others in the literature) are restricted to particular classes of systems (e.g.\ multi-affine~\cite{gol2013time}, mixed-monotone~\cite{coogan2015mixed}) and subsets of LTL formulas (e.g.\ reach-avoid-stay~\cite{nilsson2014incremental,yang2017fuel}, co-safe LTL~\cite{gol2013time}).
However, providing a framework capable of solving the synthesis problem for any dynamical systems under general LTL formulas remains a challenging problem.

This can particularly be observed when considering the main software toolboxes in the literature aimed at addressing such high-level control problems on dynamical systems, which can be split in two categories.
On one side are tools such as TuLiP~\cite{tulip}, conPAS2~\cite{conpas2} and LTLMoP~\cite{ltlmop} which can handle general LTL specifications (conPAS2) or the large subset of GR(1) formulas (TuLiP, LTLMoP) but are restricted to simpler dynamical systems such as fully actuated (LTLMoP) and piecewise affine models (TuLiP, conPAS2).
On the other side, switched or nonlinear dynamical systems are handled by tools such as PESSOA~\cite{pessoa}, CoSyMa~\cite{cosyma} and SCOTS~\cite{scots}, 
but only for combination of safety and reachability specifications.

To overcome these limitations, in this paper we propose a 3-layer hierarchical decomposition of the control problem aimed at addressing general control synthesis for nonlinear dynamical systems under LTL specifications.
As opposed to the 2-step bottom-up symbolic control approach presented above which starts by computing an abstraction of the dynamical system before synthesizing a controller on this abstraction, we rather take inspiration from top-down hierarchical decomposition in the field of artificial intelligence~\cite{russell2009modern}.
In this approach, we have several granularities of abstraction of the control problem and we first solve the problem on the most abstract layer, then iteratively refine this result by going down to a more detailed layer whose subproblem consists of realizing the solution of the above layer.
Given an initial partition of the state space (possibly containing unsafe regions) and an LTL formula defined over a set of regions of interest, each corresponding to a single cell of this partition, the proposed hierarchical decomposition proceeds to the following three steps.
% 1) Solve the LTL planning problem on a finite transition system that only represents the regions of interest, and obtain a resulting infinite sequence of regions to visit;
% 2) Find a discrete plan in the partitioned state space connecting each pair of consecutive regions in this sequence while avoiding unsafe regions;
% 3) Synthesize a controller for the dynamical system to follow each discrete plan in a sampled-time manner.
\begin{enumerate}
\item Solve the LTL planning problem on a finite transition system that only represents the regions of interest, and obtain a resulting infinite sequence of regions to visit.
\item Find a discrete plan in the partitioned state space connecting each pair of consecutive regions in this sequence while avoiding unsafe regions.
\item Synthesize a controller for the dynamical system to follow each discrete plan in a sampled-time manner.
\end{enumerate}

The main contribution of this paper is the $3$-layer hierarchical structure allowing to tackle control problems for nonlinear systems under general LTL specifications, without restricting each layer to specific tools.
The first two layers can easily be solved by classical methods for LTL model checking on finite systems~\cite{baier2008principles} and graph searches~\cite{cormen2009introduction}, respectively.
For the third layer, we consider the recent abstraction refinement approach in~\cite{meyer2017nahs} which is applicable to any system associated with a method to over-approximate its finite-time reachable sets.
The second contribution of this paper is thus the definition of a new reachability analysis method relying on the monotonicity property~\cite{smith_monotone} but applicable to any continuously differentiable system without any monotonicity assumption.

This paper is structured as follows.
In Section~\ref{sec related work}, our main contributions are compared to existing work in the literature.
Section~\ref{sec preliminaries} formulates the control problem and introduces the 3-layer hierarchical decomposition of LTL control problems on nonlinear dynamical systems.
The new reachability analysis for any continuously differentiable system is presented in Section~\ref{sec refinement} alongside an overview of the abstraction refinement algorithm in which it is used.
Finally, Section~\ref{sec simulation} presents a numerical implementation of the proposed approach to a motion planning problem for a unicycle robot with disturbances.

%%%%%%%%%%%%%%%%%%%%%%%%%%%%%%%%%%%%%%%%%%%%%%%%%%%%%%%%%%%%%%%%%%%%%%%%%%%%%%%%%%%%%%%%%%%%%%%%%%%%%%%%%%%%
%%%%%%%%%%%%%%%%%%%%%%%%%%%%%%%%%%%%%%%%%%%%%%%%%%%%%%%%%%%%%%%%%%%%%%%%%%%%%%%%%%%%%%%%%%%%%%%%%%%%%%%%%%%%
\section{Related work}
\label{sec related work}
The abstraction method in the third step uses a finite-time reachability analysis of the dynamical system to compute the non-deterministic transitions of the abstraction.
In this paper, we propose a new reachability analysis approach relying on a monotonicity property~\cite{smith_monotone} but applicable to dynamical systems which are not monotone.
A first abstraction-based control approach relying on the monotonicity property was introduced in~\cite{moor2002abstraction} for monotone systems and then extended in~\cite{coogan2015mixed} for the larger class of mixed-monotone systems.
% , which can be decomposed into their increasing and decreasing components defining a monotone \emph{decomposition function}.
An extension of the sufficient conditions for mixed-monotonicity from~\cite{coogan2015mixed} to any continuously differentiable system was recently introduced in~\cite{yang2017note} for another type of abstraction~\cite{yang2017fuel}.
Starting from systems satisfying the mild conditions in~\cite{yang2017note}, our contribution is to define a new finite-time reachability analysis approach inspired by, yet strictly more general than, the one in~\cite{coogan2015mixed}, thus opening the use of monotonicity-based abstraction approaches to any continuously differentiable system.

To further compare the proposed approach with the previously mentioned works, we can first note that the introduction of the intermediate layer in our hierarchical decomposition allows the consideration of more general control objectives than PESSOA~\cite{pessoa}, CoSyMa~\cite{cosyma} or SCOTS~\cite{scots} by translating a general LTL specification into a sequence of reachability problems.
In addition, the new contribution on monotonicity-based reachability analysis for any continuously differentiable nonlinear system opens this approach to a much wider class of systems than those considered in TuLiP~\cite{tulip}, conPAS2~\cite{conpas2} and LTLMoP~\cite{ltlmop}.
Regarding other tools also covering nonlinear systems, it should be noted that PESSOA~\cite{pessoa} does not natively handle those systems and requires the user to manually provide a Matlab function computing an over-approximation of the reachable set, while an over-approximation method is included by default in our approach.
The consideration of nonlinear systems in CoSyMa~\cite{cosyma} is based on an incremental stability assumption, which is relaxed in this paper.
Finally, SCOTS~\cite{scots} simply uses a different over-approximation method based on Lipschitz arguments to create a growth bound on the nonlinear system's reachable set.
Although we compare our contributions to existing tools as they are good indicators of the generality of results that can be covered in this field, this paper mainly focuses on providing the initial theoretical results and structure for a possible future development of a general and fully reusable tool.

Among other relevant work, 2-layer top-down structures are proposed in~\cite{kress2009temporal,wongpiromsarn2009receding} for fully actuated and piecewise affine systems respectively, where the first layer of the present paper is skipped to look directly for a discrete plan satisfying the LTL formula in the partitioned environment, then a continuous controller is designed to realize this discrete plan.
Similarly, the 3-layer top-down decomposition mentioned in~\cite{belta2007symbolic} also skips the first layer but splits the second one in two components: first finding all discrete plans satisfying the LTL formula in the partitioned environment without obstacle, then picking an optimal plan based on obstacle avoidance.
The third layer in~\cite{belta2007symbolic} uses a deterministic abstraction approach similar to~\cite{gol2013time} to implement this discrete plan on a robot modeled by an affine system.
% Among other relevant work, a 2-layer top-down structure is proposed in~\cite{fainekos2006translating} where the LTL control problem is decomposed into a hybrid automaton with safety and reachability sub-objectives on each of its nodes, to be solved by the continuous system.
Another 3-step hierarchical decomposition of an LTL control problem is presented in~\cite{fainekos2007hierarchical}, but for a bottom-up decomposition whose steps are significantly different from our approach as they consist in first abstracting the dynamical system into a fully actuated model (similarly to our second step), then robustifying the specification to compensate for the mismatches with the initial system and finally solving the new LTL problem on the robust specification.

As a summary, the main theoretical contributions of the proposed framework compared to existing tools for abstraction-based synthesis and other multi-layer approaches is the ability to handle general LTL control problems (layer $1$) on any nonlinear system (new result on reachability analysis combined with abstraction refinement in layer $3$) at a reduced computational cost (decomposing the LTL planning in two steps with layer $2$).
Such general problems cannot be handled by any of the approaches mentioned above.

%%%%%%%%%%%%%%%%%%%%%%%%%%%%%%%%%%%%%%%%%%%%%%%%%%%%%%%%%%%%%%%%%%%%%%%%%%%%%%%%%%%%%%%%%%%%%%%%%%%%%%%%%%%%
%%%%%%%%%%%%%%%%%%%%%%%%%%%%%%%%%%%%%%%%%%%%%%%%%%%%%%%%%%%%%%%%%%%%%%%%%%%%%%%%%%%%%%%%%%%%%%%%%%%%%%%%%%%%
\section{Hierarchical decomposition of LTL problems}
\label{sec preliminaries}
Let $\N$, $\R$, $\R_0^+$ and $\R_0^-$ be the sets of positive integers, reals, non-negative reals and non-positive reals, respectively.
For $a,b\in\R^n$, the interval $[a,b]\subseteq\R^n$ is defined as $[a,b]=\{z\in\R^n~|~a\leq z\leq b\}$ using componentwise inequalities.

%%%%%%%%%%%%%%%%%%%%%%%%%%%%%%%%%%%%%%%%%%%%%%%%%%%%%%
\subsection{System description}
\label{sub prelim system}
We consider a class of continuous-time nonlinear control systems subject to disturbances and modeled by:
\begin{equation}
\label{eq system}
\dot z = f(z,u,d),
\end{equation}
where $z\in\mathcal{Z}\subseteq\R^n$, $u\in\mathcal{U}\subseteq\R^p$ and $d\in\mathcal{D}\subseteq\R^q$ are the state, bounded control input and bounded disturbance input, respectively.
Throughout this paper, the vector field $f$ of (\ref{eq system}) is assumed to be continuously differentiable.
We denote as $\Phi(t,z,\mathbf{u},\mathbf{d})$ the state (assumed to exist and be unique) reached by (\ref{eq system}) at time $t\in\R_0^+$ from initial state $z\in\mathcal{Z}$, under the piecewise continuous control $\mathbf{u}:\mathbb{R}_0^+\rightarrow\mathcal{U}$ and disturbance functions $\mathbf{d}:\mathbb{R}_0^+\rightarrow\mathcal{D}$.
We use $\Phi(t,z,u,d)$ with $u\in\mathcal{U}$ and $d\in\mathcal{D}$ in the case of constant input functions $\mathbf{u}:\mathbb{R}_0^+\rightarrow\{u\}$ and $\mathbf{d}:\mathbb{R}_0^+\rightarrow\{d\}$.

Given a sampling period $\tau\in\R_0^+$, a sampled version of system (\ref{eq system}) can be described as a non-deterministic (due to the disturbance) infinite transition system $S_\tau=(X_\tau,U_\tau,\delta_\tau)$ where: 
\ifdouble
  $X_\tau=\mathcal{Z}$ is the set of states;
  $U_\tau=\mathcal{U}$ is the set of control inputs;
  the transition relation $\delta_\tau:X_\tau\times U_\tau\rightarrow X_\tau$ is such that $z'\in\delta_\tau(z,u)$ if there exists a disturbance $\mathbf{d}:[0,\tau]\rightarrow\mathcal{D}$ such that $z'=\Phi(\tau,z,u,\mathbf{d})$, i.e.\ $z'$ can be reached from $z$ exactly in time $\tau$ by applying the constant control $u$ on $[0,\tau]$.
\else
  \begin{itemize}
  \item $X_\tau=\mathcal{Z}$ is the set of states,
  \item $U_\tau=\mathcal{U}$ is the set of control inputs,
  \item the transition relation $\delta_\tau:X_\tau\times U_\tau\rightarrow X_\tau$ is such that $z'\in\delta_\tau(z,u)$ if there exists a disturbance $\mathbf{d}:[0,\tau]\rightarrow\mathcal{D}$ such that $z'=\Phi(\tau,z,u,\mathbf{d})$, i.e.\ $z'$ can be reached from $z$ exactly in time $\tau$ by applying the constant control $u$ on $[0,\tau]$.
  \end{itemize}
\fi
While, to the best of our knowledge, there exists very few results involving the choice of the sampling period for abstraction-based approaches~\cite{boskos2015decentralized}, some guidelines are provided in Section~\ref{subsub simu layer3} for a unicycle model and in~\cite{meyer2017ifac} for systems with additive control input ($\dot z=f(z,d)+u$).

%%%%%%%%%%%%%%%%%%%%%%%%%%%%%%%%%%%%%%%%%%%%%%%%%%%%%%%%%%%%%%%%%%%%%%%%%%%%%%%%%%%%%%%%%%%%%%%%%%%%%%%%%%%%
\subsection{Hierarchical decomposition of an LTL control problem}
\label{sub prelim hierarchical}
We consider a high-level control problem on the sampled-time system $S_\tau$ evolving in the workspace $X_\tau=\mathcal{Z}\subseteq\R^n$ associated to a uniform partition $\mathcal{P}\subseteq 2^\mathcal{Z}$ into intervals (for compatibility with the reachability analysis introduced in Section~\ref{sub ar reachability}).
The description of this workspace also includes a set $Obs\subseteq\mathcal{P}$ of unsafe regions (referred to as \emph{obstacles} in this section) and a set $\Pi\subseteq\mathcal{P}\backslash Obs$ of \emph{regions of interest}.
The control specification is described by a Linear Temporal Logic (LTL) formula $\varphi$ defined over the set of regions of interest $\Pi$.
The reader is referred to~\cite{baier2008principles} for an introduction on the LTL framework.
We thus aim at solving the following problem.
\begin{problem}
\label{pb LTL}
Find a controller $C:X_\tau\rightarrow U_\tau$ such that the closed-loop sampled system $S_\tau$ with transitions $\delta_\tau(z,C(z))$ satisfies the LTL formula $\varphi$ while avoiding the obstacles $Obs\subseteq\mathcal{P}$.
\end{problem}

To solve Problem~\ref{pb LTL}, we propose a hierarchical control structure involving three different abstraction layers of the dynamical system and its environment, each of which successively addresses one aspect of the control problem as sketched in Figure~\ref{fig hierarchy}.
The evolution of the control objectives (highlighted in red in Figure~\ref{fig hierarchy}) is obtained through the following three steps, each applied on a different abstraction layer (in blue).
% \begin{itemize}
% \item We first solve the LTL planning problem on a finite transition system representing only the regions of interest (\emph{RoI} in Figure~\ref{fig hierarchy}) of the workspace. The resulting \emph{accepting path} is a (possibly infinite) sequence of regions of interest satisfying the LTL formula $\varphi$.
% \item Based on the workspace partition $\mathcal{P}$ and its subset of obstacles $Obs\subseteq\mathcal{P}$, a finite transition system describing possible motion in this workspace (disregarding the dynamics of (\ref{eq system})) taking into account obstacle avoidance is created and used to obtain a discrete plan in $\mathcal{P}$ connecting all pairs of consecutive regions of interest in the accepting path.
% \item A controller for the dynamical system to follow these discrete plans is finally synthesized through an abstraction refinement approach relying on the over-approximation operator (\ref{eq over approximation}).
% \end{itemize}

\begin{figure}[htb]
\centering
\includegraphics[width=\columnwidth]{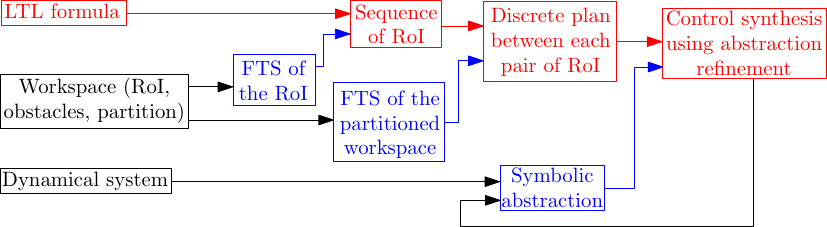}
\caption{Hierarchical structure of the problem solution (RoI $=$ \emph{Regions of Interest}, FTS $=$ \emph{Finite Transition System}).}
\label{fig hierarchy}
\end{figure}

%%%%%%%%%%%%%%%%%%%%%%%%%%%%%%%%%%%%%%%%%%%%%%%%%%%%%%
\subsubsection{LTL planning}
\label{subsub prelim LTL}
We first solve the LTL planning problem on a transition system $S_{\Pi}=(\Pi,\delta_{\Pi})$ whose states are the regions of interest in the finite set $\Pi\subseteq\mathcal{P}\backslash Obs$ and its transition relation is $\delta_{\Pi}\subseteq\Pi^2$.
The user can freely choose between defining $\delta_{\Pi}=\Pi^2$ (i.e.\ each region of interest can reach any other) to disregard the workspace geometry or manually creating $\delta_{\Pi}\varsubsetneq\Pi^2$ to consider physical constraints.
Any standard LTL model checker (see e.g.~\cite{holzmann2003spin}) can then be used with $S_\Pi$ and $\varphi$ to obtain an \emph{accepting path} in $S_\Pi$ satisfying the LTL formula $\varphi$.
Let $\bar\pi=\pi^0\pi^1\pi^2\dots$ denote this accepting path (if it exists) as a (possibly infinite) sequence of regions of interest in $\Pi$.

%%%%%%%%%%%%%%%%%%%%%%%%%%%%%%%%%%%%%%%%%%%%%%%%%%%%%%
\subsubsection{Discrete plan and obstacle avoidance}
\label{subsub prelim plan}
Next we focus on obtaining a discrete plan in the workspace partition $\mathcal{P}$ realizing the accepting path $\bar\pi$.
Let $\mathcal{N}(\sigma)\subseteq\mathcal{P}$ be the set of neighbors of a cell $\sigma\in\mathcal{P}$ (i.e.\ the partition elements having a common facet with $\sigma$), with the assumption that $\sigma\in\mathcal{N}(\sigma)$.
The second abstraction layer thus describes possible motion in the physical environment while disregarding the system dynamics and is represented by the transition system $S_{\mathcal{P}}=(\mathcal{P},\delta_{\mathcal{P}})$ whose set of states (or \emph{cells}) is $\mathcal{P}$ and its transition relation $\delta_{\mathcal{P}}\subseteq\mathcal{P}^2$ is such that $(\sigma,\sigma')\in\delta_{\mathcal{P}}$ for all $\sigma\in\mathcal{P}\backslash Obs$, $\sigma'\in\mathcal{N}(\sigma)\backslash Obs$.
As a result, any behavior of $S_\mathcal{P}$ induced by the above transition relation $\delta_\mathcal{P}$ is guaranteed to satisfy the obstacle avoidance.

Then for each pair $(\pi^i,\pi^{i+1})\in \Pi^2$ of consecutive regions of interest in the accepting path $\bar\pi$, we look for a plan $\Lambda^i=\sigma^i_0\sigma^i_1\dots\sigma^i_{r_i}$ in $\mathcal{P}\backslash Obs$ connecting the two cells $\sigma^i_0=\pi^i$ and $\sigma^i_{r_i}=\pi^{i+1}$.
Since $\Pi$ is finite, there is necessarily a finite number of such pairs appearing in an infinite accepting path $\bar\pi$, which means that even an LTL formula with infinitely repeating properties is translated into a finite set of finite plans.
The search for the plan $\Lambda^i=\sigma^i_0\sigma^i_1\dots\sigma^i_{r_i}$ is done through classical graph search algorithms on $S_{\mathcal{P}}$ (see e.g.~\cite{cormen2009introduction}), such as a Breadth-First Search or a Dijkstra algorithm to consider weighted transitions (e.g.\ to penalize transitions going to a cell neighboring an obstacle), which are guaranteed to find such plans $\Lambda^i$ as long as $\pi^i$ and $\pi^{i+1}$ can be connected in $S_\mathcal{P}$.

\begin{remark}
\label{rmk roi}
$\Pi\subseteq\mathcal{P}$ ensures that each pair $(\pi^i,\pi^{i+1})$ in $\bar\pi$ only needs one plan $\Lambda^i$ as above.
The more general case with $\Pi\subseteq 2^\mathcal{P}$ can be handled at a greater computational cost, since a plan in $\mathcal{P}$ needs to be found for each pair in $\pi^i\times\pi^{i+1}\subseteq\mathcal{P}^2$.
\end{remark}

%%%%%%%%%%%%%%%%%%%%%%%%%%%%%%%%%%%%%%%%%%%%%%%%%%%%%%
\subsubsection{Control synthesis}
\label{subsub prelim synthesis}
After applying the first two layers as above, Problem~\ref{pb LTL} can then be solved by obtaining a solution to the following problem for each plan $\Lambda^i$ from Section~\ref{subsub prelim plan}.
\begin{problem}
\label{pb refinement}
Given a plan $\Lambda^i=\sigma^i_0\sigma^i_1\dots\sigma^i_{r_i}$ in $\mathcal{P}$, find a controller $C^i:X_\tau\rightarrow U_\tau$ such that the closed-loop sampled system $S_\tau$ follows this plan, i.e.\ for any trajectory $z_0\dots z_{r_i}$ of $S_\tau$ with $z_0\in\sigma^i_0$ and $z_{k+1}\in\delta_\tau(z_k,C^i(z_k))$ for all $k\in\{0,\dots,r_i-1\}$, it holds that $z_k\in\sigma^i_k$.
\end{problem}

As for the other layers, any tool able to solve Problem~\ref{pb refinement} can be used in this third layer.
One possible example is the approach in~\cite{gol2013time} where an abstraction of $S_\tau$ is created as a finite transition system obtained by designing feedback controllers deterministically driving the system between any two neighbor cells in $\mathcal{P}$.
While the obtained deterministic abstraction leads to a straightforward solution for Problem~\ref{pb refinement}, such approach is limited to the class of multi-affine systems.

For the highest generality of the proposed hierarchical structure, the suggested solution for the third layer (detailed in Section~\ref{sec refinement}) is another abstraction-based approach relying on reachability analysis of system (\ref{eq system}) and which has the advantage of being applicable to any continuously differentiable system.
Unlike~\cite{gol2013time} above, the drawback of such approach is that the obtained abstraction is a non-deterministic transition system (see e.g.~\cite{coogan2015mixed}) which is thus unlikely to deterministically follow the plan $\Lambda^i$ in Problem~\ref{pb refinement}.
Instead of manually looking for abstractions of finer granularity on which the control problem is feasible, we consider an abstraction refinement approach as in~\cite{meyer2017nahs} where an abstraction is created on the initial coarse partition $\mathcal{P}$ and then iteratively refined by re-partitioning the cells where the control synthesis fails.

%%%%%%%%%%%%%%%%%%%%%%%%%%%%%%%%%%%%%%%%%%%%%%%%%%%%%%
\subsubsection{General comments}
\label{subsub prelim comments}
The main contribution of this section is the new $3$-layer hierarchical framework to solve Problem~\ref{pb LTL} for nonlinear systems under general LTL specifications.
However as mentioned above, each layer can be tackled by any existing tool able to address the corresponding subproblem, and the examples mentioned in this section are not claimed to be the optimal choices nor to be new contributions of this paper.
Unlike the first two layers relying on well established tools, the abstraction refinement approach proposed for the third layer is a more recent result~\cite{meyer2017nahs}, and although its algorithm is not new, its applicability to any continuously differentiable system is the second main contribution of this paper which we thus describe in more detail in Section~\ref{sec refinement}.

The proposed structure combining the first and second layers has the advantage of providing a solution to the LTL planning on $\mathcal{P}$ at a significantly lower computational cost than if this problem were to be solved in a single step, directly on $S_{\mathcal{P}}$.
As will be seen in the numerical example of Section~\ref{sec simulation}, the main computational bottleneck is on the control synthesis in the third layer, which is the trade-off for the generality offered by the proposed abstraction refinement approach.
As mentioned above, the computational complexity of layer $3$ can be reduced at the cost of generality by using more efficient tools which are only applicable to smaller classes of systems.

Finally, we provide some guidelines on how to handle infeasibility of each layer's subproblem.
In layer $1$, if the LTL specification $\varphi$ is infeasible, \emph{specification revision} methods~\cite{kim2015minimal} can be considered to find a new specification satisfiable by $S_\Pi$ and as close to $\varphi$ as possible.
In layer $2$, if regions of interest $\pi^i,\pi^j\in \Pi$ cannot be connected in $\mathcal{P}\setminus Obs$, the first abstraction layer $S_\Pi$ needs to be updated with the physical constraints: $(\pi^i,\pi^j)\notin\delta_\Pi$ and $(\pi^j,\pi^i)\notin\delta_\Pi$.
In layer $3$, if we fail to synthesize a controller within reasonable refinement iterations, the abstraction refinement algorithm can be combined with the \emph{plan revision} approach in~\cite{meyer2017ifac} to look for alternative plans $\Lambda^i$.

%%%%%%%%%%%%%%%%%%%%%%%%%%%%%%%%%%%%%%%%%%%%%%%%%%%%%%%%%%%%%%%%%%%%%%%%%%%%%%%%%%%%%%%%%%%%%%%%%%%%%%%%%%%%
%%%%%%%%%%%%%%%%%%%%%%%%%%%%%%%%%%%%%%%%%%%%%%%%%%%%%%%%%%%%%%%%%%%%%%%%%%%%%%%%%%%%%%%%%%%%%%%%%%%%%%%%%%%%
\section{Generalized abstraction refinement}
\label{sec refinement}
The abstraction refinement approach considered in the third layer and initiated in~\cite{meyer2017nahs} is applicable to any system (\ref{eq system}) whose finite-time reachable sets can be efficiently over-approximated.
In Section~\ref{sub ar reachability}, we thus introduce the second main contribution of this paper: a new reachability analysis approach applicable to any continuously differentiable system (\ref{eq system}).
For self-containment of the paper, Section~\ref{sub ar refinement} then provides an overview of the abstraction refinement algorithm from~\cite{meyer2017nahs}.

%%%%%%%%%%%%%%%%%%%%%%%%%%%%%%%%%%%%%%%%%%%%%%%%%%%%%%
\subsection{General monotonicity-based reachability analysis}
\label{sub ar reachability}
Monotone systems are systems whose trajectories preserve some partial orders as below.
A formal definition of a partial order is omitted in this paper but can be found in~\cite{angeli_monotone}.
\begin{definition}
\label{def monotone}
System (\ref{eq system}) is monotone with respect to partial orders $\preceq_z$, $\preceq_u$ and $\preceq_d$ on the state, control and disturbance inputs respectively, if for all time $t\in\R_0^+$, initial states $z,z'\in\mathcal{Z}$, control functions $\mathbf{u},\mathbf{u}':[0,t]\rightarrow\mathcal{U}$ and disturbance functions $\mathbf{d},\mathbf{d}':[0,t]\rightarrow\mathcal{D}$ we have:
\ifdouble
  ${z\preceq_z z'}, {\mathbf{u}\preceq_u \mathbf{u}'}, {\mathbf{d}\preceq_d \mathbf{d}'}\Rightarrow \Phi(t,z,\mathbf{u},\mathbf{d})\preceq_z\Phi(t,z',\mathbf{u}',\mathbf{d}')$.
%   \begin{IEEEeqnarray*}{l}
%   z\preceq_z z',\mathbf{u}\preceq_u \mathbf{u}',\mathbf{d}\preceq_d \mathbf{d}'\\
%   \qquad\qquad\qquad\Rightarrow \Phi(t,z,\mathbf{u},\mathbf{d})\preceq_z\Phi(t,z',\mathbf{u}',\mathbf{d}').
%   \end{IEEEeqnarray*}
\else
  $$z\preceq_z z',\mathbf{u}\preceq_u \mathbf{u}',\mathbf{d}\preceq_d \mathbf{d}'\Rightarrow \Phi(t,z,\mathbf{u},\mathbf{d})\preceq_z\Phi(t,z',\mathbf{u}',\mathbf{d}').$$
\fi
\end{definition}

In this section, we provide a new reachability analysis approach relying on the monotonicity property but without any monotonicity assumption on (\ref{eq system}).
The recent results in~\cite{yang2017note} extend the sufficient conditions for mixed-monotonicity (see e.g.~\cite{coogan2015mixed}) to any continuous-time system whose Jacobian matrices are bounded over the considered sets of states and inputs.
Our new contribution in this section starts from systems satisfying the very mild sufficient conditions in~\cite{yang2017note} and defines a new reachability analysis approach which is inspired by but is strictly more general than the one in~\cite{coogan2015mixed}.
% Below, we thus describe how this extended characterization of mixed-monotonicity can be exploited to compute a finite-time over-approximation of reachable sets for any continuously differentiable continuous-time dynamical system.

In what follows, several steps need to be identically applied to all three variables $z\in\mathcal{Z}\subseteq\R^n$, $u\in\mathcal{U}\subseteq\R^p$ and $d\in\mathcal{D}\subseteq\R^q$.
When this is the case, we will use generic notations with variable $c\in\{z,u,d\}$ and dimension $m\in\{n,p,q\}$ such that $c\in\R^m$.
We first denote as $a^c_{ij}$ and $b^c_{ij}$ the bounds of the partial derivatives of the vector field $f$ as follows: for all $z\in\mathcal{Z}$, $u\in\mathcal{U}$, $d\in\mathcal{D}$, $i\in\{1,\dots,n\}$ and $j\in\{1,\dots,m\}$,
$\frac{\partial f_i}{\partial c_j}(z,u,d)\in[a^c_{ij},b^c_{ij}]$.
% \begin{equation}
% \label{eq bounds}
% \frac{\partial f_i}{\partial c_j}(z,u,d)\in[a^c_{ij},b^c_{ij}].
% \end{equation}
The values of these bounds lead us to consider the 4 cases below, covering all possibilities for the sign of each partial derivative, as in~\cite{yang2017note}:
\begin{enumerate}[leftmargin=3em]
\renewcommand{\labelenumi}{(C\arabic{enumi})}
\item $a^c_{ij}\geq 0$: positive,
\item $a^c_{ij}\leq 0\leq b^c_{ij}$ and $|a^c_{ij}|\leq|b^c_{ij}|$: mostly positive,
\item $a^c_{ij}\leq 0\leq b^c_{ij}$ and $|a^c_{ij}|\geq|b^c_{ij}|$: mostly negative,
\item $b^c_{ij}\leq 0$: negative.
\end{enumerate}

We then define the function $g$ such that for all $z,z^*\in\mathcal{Z}$, $u,u^*\in\mathcal{U}$, $d,d^*\in\mathcal{D}$ and $i\in\{1,\dots,n\}$ we have
\ifdouble
  \begin{IEEEeqnarray}{rl}
  g_i(z,u,d,z^*,u^*,d^*)=&~f_i(Z_i,U_i,D_i) + \alpha_i^z(z-z^*)\label{eq decomposition}\\
  &+~\alpha_i^u(u-u^*) + \alpha_i^d(d-d^*),\nonumber
  \end{IEEEeqnarray}
\else
  \begin{equation}
  \label{eq decomposition}
  g_i(z,u,d,z^*,u^*,d^*)=f_i(Z_i,U_i,D_i) + \alpha_i^z(z-z^*) + \alpha_i^u(u-u^*) + \alpha_i^d(d-d^*),
  \end{equation}
\fi
where the components of $Z_i=(z_{i1} \dots z_{in})^\top$, $U_i=(u_{i1} \dots u_{ip})^\top$, $D_i=(d_{i1} \dots d_{iq})^\top$ and $\alpha_i^c=(\alpha^c_{i1} \dots \alpha^c_{im})$ are defined according to cases (C1)-(C4) for $\frac{\partial f_i}{\partial c_j}$ as follows with $c\in\{z,u,d\}$ in all notations below:
\ifdouble
  $$c_{ij}=\begin{cases}
  c_j & \text{if (C1) or (C2),}\\
  c_j^* & \text{if (C3) or (C4),}\end{cases}
  \quad
  \alpha^c_{ij}=\begin{cases}
  -a^c_{ij} & \text{if (C2),}\\
  b^c_{ij} & \text{if (C3),}\\
  0 & \text{otherwise.}\end{cases}$$
\else
  $$c_{ij}=\begin{cases}
  c_j & \text{if (C1) or (C2),}\\
  c_j^* & \text{if (C3) or (C4),}\end{cases}
  \qquad\text{and}\qquad
  \alpha^c_{ij}=\begin{cases}
  -a^c_{ij} & \text{if (C2),}\\
  b^c_{ij} & \text{if (C3),}\\
  0 & \text{otherwise.}\end{cases}$$
\fi

The above definition of $g$ is a straightforward extension, to non-autonomous systems with control and disturbance inputs, of the one introduced in~\cite{yang2017note}.
In what follows, we provide our main contribution on this topic which describes how to compute an interval over-approximation of the finite-time reachable set for any continuously differentiable system (\ref{eq system}), without needing any additional assumption.
For this we define the following dynamical system evolving in $\mathcal{Z}^2$:
\begin{equation}
\label{eq system duplicated}
\begin{pmatrix}
\dot z\\ \dot z^*
\end{pmatrix}
=h(z,u,d,z^*,u^*,d^*)=
\begin{pmatrix}
g(z,u,d,z^*,u^*,d^*)\\g(z^*,u^*,d^*,z,u,d)
\end{pmatrix}
\end{equation}
and similarly to $\Phi$ we denote the trajectories of (\ref{eq system duplicated}) as $\Phi_h(\cdot,z,\mathbf{u},\mathbf{d},z^*,\mathbf{u^*},\mathbf{d^*}):\R_0^+\rightarrow\mathcal{Z}^2$, where bold variables are piecewise continuous input functions.
Let $\Phi_h^1$ and $\Phi_h^2$ denote the first $n$ and last $n$ components of $\Phi_h$, respectively.
Then, a single successor of (\ref{eq system duplicated}) can be used to compute an over-approximation of the finite-time reachable set of (\ref{eq system}) as follows.
\begin{theorem}
\label{th over approximation}
For all bounds $\underline{z},\overline{z}\in\R^n$, $\underline{u},\overline{u}\in\R^p$, $\underline{d},\overline{d}\in\R^q$ and for all $t\in\R_0^+$, $z\in[\underline{z},\overline{z}]$, $\mathbf{u}:[0,t]\rightarrow[\underline{u},\overline{u}]$ and $\mathbf{d}:[0,t]\rightarrow[\underline{d},\overline{d}]$ we have (using componentwise inequalities):
$$\Phi_h^1(t,\underline{z},\underline{u},\underline{d},\overline{z},\overline{u},\overline{d})\leq
\Phi(t,z,\mathbf{u},\mathbf{d})\leq
\Phi_h^2(t,\underline{z},\underline{u},\underline{d},\overline{z},\overline{u},\overline{d}).$$
\end{theorem}
\begin{proof}
From the terms $\alpha_i^c(c-c^*)$ in (\ref{eq decomposition}), we can show similarly to~\cite{yang2017note} that for all variables $c\in\{z,u,d\}$ (with $c\in\R^m$, $m\in\{n,p,q\}$), $i\in\{1,\dots,2n\}$ and $j\in\{1,\dots,m\}$ we have
$$
\frac{\partial h_i}{\partial c_j}
\begin{cases}
\geq 0 & \text{if } i\leq n,\\
\leq 0 & \text{if } i\geq n,
\end{cases}
\qquad\text{and}\qquad
\frac{\partial h_i}{\partial c_j^*}
\begin{cases}
\leq 0 & \text{if } i\leq n,\\
\geq 0 & \text{if } i\geq n.
\end{cases}
$$

Then from~\cite{angeli_monotone}, (\ref{eq system duplicated}) is monotone as in Definition~\ref{def monotone} with the partial orders $\preceq_z$, $\preceq_u$ and $\preceq_d$ on the spaces $\R^{2n}$, $\R^{2p}$ and $\R^{2q}$, respectively, as defined below.
For all variables $c\in\{z,u,d\}$ with $c\in\R^m$, the partial order $\preceq_c$ is characterized by the orthant $(\R_0^+)^m\times(\R_0^-)^m$ of space $\R^{2m}$ as follows: 
for all $c^1,c^2,c^3,c^4\in\R^m$, 
$\begin{pmatrix}c^1 \\ c^2\end{pmatrix}
\preceq_c
\begin{pmatrix}c^3 \\ c^4\end{pmatrix}
\Leftrightarrow
\begin{cases}
c^1\leq c^3,\\
c^2\geq c^4,
\end{cases}$
% \begin{equation}
% \label{eq partial order}
% \forall c^1,c^2,c^3,c^4\in\R^m,~
% \begin{pmatrix}c^1 \\ c^2\end{pmatrix}
% \preceq_c
% \begin{pmatrix}c^3 \\ c^4\end{pmatrix}
% \Leftrightarrow
% \begin{cases}
% c^1\leq c^3,\\
% c^2\geq c^4,
% \end{cases}
% \end{equation}
where $\leq$ and $\geq$ are the componentwise inequalities on $\R^m$.
For all $c\in[\underline{c},\overline{c}]\subseteq\R^m$ we thus have
$\begin{pmatrix}\underline{c} \\ \overline{c}\end{pmatrix}
\preceq_c
\begin{pmatrix}c \\ c\end{pmatrix}
\preceq_c
\begin{pmatrix}\overline{c} \\ \underline{c}\end{pmatrix}$
and we can then use Definition~\ref{def monotone} for system (\ref{eq system duplicated}) to obtain the following over-approximation: for all $t\in\R_0^+$, $z\in[\underline{z},\overline{z}]$, $\mathbf{u}:[0,t]\rightarrow[\underline{u},\overline{u}]$ and $\mathbf{d}:[0,t]\rightarrow[\underline{d},\overline{d}]$ we have
\ifdouble
  $$\begin{cases}
  \Phi_h(t,\underline{z},\underline{u},\underline{d},\overline{z},\overline{u},\overline{d})\preceq_z\Phi_h(t,z,\mathbf{u},\mathbf{d},z,\mathbf{u},\mathbf{d}),\\
  \Phi_h(t,z,\mathbf{u},\mathbf{d},z,\mathbf{u},\mathbf{d})\preceq_z\Phi_h(t,\overline{z},\overline{u},\overline{d},\underline{z},\underline{u},\underline{d}).
  \end{cases}$$
\else
  $$\Phi_h(t,\underline{z},\underline{u},\underline{d},\overline{z},\overline{u},\overline{d})\preceq_z
  \Phi_h(t,z,\mathbf{u},\mathbf{d},z,\mathbf{u},\mathbf{d})\preceq_z
  \Phi_h(t,\overline{z},\overline{u},\overline{d},\underline{z},\underline{u},\underline{d}).$$
\fi
From (\ref{eq decomposition}), $g(z,u,d,z,u,d)=f(z,u,d)$ which implies $\Phi_h(t,z,\mathbf{u},\mathbf{d},z,\mathbf{u},\mathbf{d})=\begin{pmatrix}\Phi(t,z,\mathbf{u},\mathbf{d}) \\ \Phi(t,z,\mathbf{u},\mathbf{d})\end{pmatrix}$.
Then by symmetry of (\ref{eq system duplicated}), we have $\Phi_h^1(t,\overline{z},\overline{u},\overline{d},\underline{z},\underline{u},\underline{d})=\Phi_h^2(t,\underline{z},\underline{u},\underline{d},\overline{z},\overline{u},\overline{d})$ finally giving the result in Theorem~\ref{th over approximation}.
\end{proof}

Theorem~\ref{th over approximation} thus provides a method to obtain over-approximations of the finite-time reachable sets for any continuously differentiable system (\ref{eq system}) by computing a single successor state $\Phi_h(t,\underline{z},\underline{u},\underline{d},\overline{z},\overline{u},\overline{d})$ of system (\ref{eq system duplicated}).

%%%%%%%%%%%%%%%%%%%%%%%%%%%%%%%%%%%%%%%%%%%%%%%%%%%%%%
\subsection{Abstraction refinement algorithm}
\label{sub ar refinement}
For each plan $\Lambda^i$ obtained in the second layer (Section~\ref{subsub simu layer2}), we abstract the sampled system $S_\tau$ by the finite transition system $S_a^i=(X_a^i,U_a,\delta_a^i)$, where:
\ifdouble
  the set of states (or \emph{symbols}) $X_a^i$ is a partition of the workspace $\mathcal{Z}\subseteq\R^n$ into intervals, i.e.\ any symbol $s\in X_a^i$ is also an interval $s=[\underline{s},\overline{s}]\subseteq\mathcal{Z}$ of the workspace;
  the set of inputs $U_a$ is a finite subset of control values in $\mathcal{U}$;
  a transition $s'\in\delta_a^i(s,u)$ between symbols $s\in X_a^i$ and $s'\in X_a^i$ with input $u\in U_a$ exists if $s'\cap[\Phi_h^1(\tau,\underline{s},u,\underline{d},\overline{s},u,\overline{d}),\Phi_h^2(\tau,\underline{s},u,\underline{d},\overline{s},u,\overline{d})]\neq\emptyset$.
\else
  \begin{itemize}
  \item the set of states (or \emph{symbols}) $X_a^i$ is a partition of the workspace $\mathcal{Z}\subseteq\R^n$ into intervals, i.e.\ any symbol $s\in X_a^i$ is also an interval $s=[\underline{s},\overline{s}]\subseteq\mathcal{Z}$ of the workspace,
  \item the set of inputs $U_a$ is a finite subset of control values in $\mathcal{U}$,
  \item a transition $s'\in\delta_a^i(s,u)$ between symbols $s\in X_a^i$ and $s'\in X_a^i$ with input $u\in U_a$ exists if $s'\cap[\Phi_h^1(\tau,\underline{s},u,\underline{d},\overline{s},u,\overline{d}),\Phi_h^2(\tau,\underline{s},u,\underline{d},\overline{s},u,\overline{d})]\neq\emptyset$.
  \end{itemize}
\fi

\begin{remark}
\label{rmk bounds}
The over-approximation of the reachable set of (\ref{eq system}) from Theorem~\ref{th over approximation} with the constant control value $u$ used above to define $\delta_a^i$ can also be obtained with the less conservative local bounds for the Jacobians $\frac{\partial f}{\partial z}$ and $\frac{\partial f}{\partial d}$ over the set of possible states only during the time period $[0,\tau]$.
\end{remark}

Since the transition relation $\delta_a^i$ is non-deterministic and the objective of Problem~\ref{pb refinement} is to deterministically follow the plan $\Lambda^i$ in $\mathcal{P}$, we rely on an abstraction refinement approach by considering the initial coarse partition $X_a^i=\mathcal{P}$ which is then iteratively refined by re-partitioning the elements of $X_a^i$ that are responsible for preventing the synthesis of a controller.
Below, we provide an overview of the considered refinement algorithm.
For further details, the reader is referred to~\cite{meyer2017nahs}.

We first define the function $P_a^i:\mathcal{P}\rightarrow 2^{X_a^i}$ such that $P_a^i(\sigma)=\{s\in X_a^i~|~s\subseteq\sigma\}$ corresponds to the projection of a cell $\sigma\in \mathcal{P}$ onto the given finer partition $X_a^i$.
Then, Algorithm~\ref{algo global} is centered around the computation of \emph{valid sets} defined as a function $V^i:\mathcal{P}\rightarrow 2^{X_a^i}$ by proceeding backwards along the plan $\Lambda^i=\sigma^i_0\sigma^i_1\dots\sigma^i_{r_i}$.
The final cell $\sigma^i_{r_i}$ of $\Lambda^i$ is considered as valid and the valid set function is thus initialized with $V^i(\sigma^i_{r_i})=\{\sigma^i_{r_i}\}$.
Other valid sets $V^i(\sigma^i_k)$ are then iteratively defined as the subset of symbols in $\sigma^i_k$ which can be driven towards the valid set $V^i(\sigma^i_{k+1})$ of the next cell for at least one control input in $U_a$:
$V^i(\sigma^i_k)=\left\{s\in P_a^i(\sigma^i_k)~\left|~\exists u\in U_a,~\delta_a^i(s,u)\subseteq V^i(\sigma^i_{k+1})\right.\right\}$.
The controller $C_a^i:X_a^i\rightarrow {U_a}$ is simultaneously defined by associating to each valid symbol $s\in V^i(\sigma^i_k)$ the \emph{first} of such control values.
Given a cell $\sigma^i_k$ and a targeted valid set $V^i(\sigma^i_{k+1})$, the function $\mathtt{ValidSet}(\sigma^i_k,V^i(\sigma^i_{k+1}))$ in Algorithm~\ref{algo global} denotes the computation of $V^i(\sigma^i_k)$ and $C_a^i$ as above.

\begin{algorithm2e}[tbh]
  \SetKwFunction{Pick}{Pick}
  \SetKwFunction{Split}{Split}
  \SetKwFunction{ValidSet}{ValidSet}
  \KwData{$\mathcal{P}$, $\Lambda^i=\sigma^i_0\sigma^i_1\dots\sigma^i_{r_i}\in\mathcal{P}^{r_i+1}$, $P_a^i:\mathcal{P}\rightarrow 2^{X_a^i}$.}
  {\bf Initialization:} $X_a^i=\mathcal{P}$, $V^i(\sigma^i_{r_i})=\{\sigma^i_{r_i}\}$\\
  \For{$k$ from $r_i-1$ to $0$}
  {
    $\{V^i(\sigma^i_k),~C_a^i\}=\ $\ValidSet$(\sigma^i_k,V^i(\sigma^i_{k+1}))$\\
    \While{$V^i(\sigma^i_k)=\emptyset$ or $V^i(\sigma^i_0)\neq P^i_a(\sigma^i_0)$}
    {
      $j=\Pick(k,r_i-1)$\\
      \ForAll{$s\in P_a^i(\sigma^i_j)\backslash V^i(\sigma^i_j)$}{$X_a^i=(X_a^i\backslash \{s\})\cup\Split(s)$}
      \For{$l$ from $j$ to $k$}{$\{V^i(\sigma^i_l),~C_a^i\}=\ $\ValidSet$(\sigma^i_l,V^i(\sigma^i_{l+1}))$}
    }
  }
  \KwOut{$\{X_a^i,~V^i:\mathcal{P}\rightarrow 2^{X_a^i},~C_a^i:X_a^i\rightarrow {U_a}\}$}
\caption{Abstraction refinement algorithm.\label{algo global}}
\end{algorithm2e}

If $V^i(\sigma^i_k)=\emptyset$ for some $k$ (line 4), we first pick a cell $\sigma^i_j$ with $j\in\{k,\dots,r_i-1\}$ to be refined (line 5), split each of its invalid symbols $s\in P_a^i(\sigma^i_j)\backslash V^i(\sigma^i_j)$ into a set of subsymbols $\texttt{Split}(s)$ and update the partition $X_a^i$ accordingly (lines 6-7), and finally update the valid sets and controller for all cells from $\sigma^i_j$ to $\sigma^i_k$ whose valid sets may be expanded after this refinement (lines 8-9).
The refinement procedure is repeated until $V^i(\sigma^i_0)= P_a(\sigma^i_0)$ (line 4), i.e.\ when starting from any subsymbol in $\sigma^i_0$, $S_a^i$ can be controlled to reach $\sigma^i_{r_i}$ exactly in $r_i$ steps.
For each plan $\Lambda^i$, Algorithm~\ref{algo global} then returns the refined partition $X_a^i$, the valid set function $V^i$ and the associated controller $C_a^i$.
The definition of both functions $\texttt{Pick}$ and $\texttt{Split}$ can be arbitrary but some guidelines and examples are provided in~\cite{meyer2017ifac,meyer2017acc,meyer2017nahs} and Section~\ref{subsub simu layer3}.
% The definition of both functions $\texttt{Pick}$ and $\texttt{Split}$ can be arbitrary but we can provide some guidelines.
% It is usually advised to prioritize the refinement of the coarsest of the visited cells $\sigma^i_k,\dots,\sigma^i_{r-1}$ as they are more likely to be the reason for having $V^i(\sigma^i_k)=\emptyset$.
% Possible choices of the function $\texttt{Pick}$ are described in~\cite{meyer2017acc} using a priority queue and in~\cite{meyer2017ifac} using a cost function estimating the complexity of the remaining computations.
% For $\texttt{Split}$, one should aim at obtaining subsymbols which remain compatible with the over-approximation method used to obtain $S_a^i$ (intervals of $\R^n$ for the monotonicity-based approach of this paper).
% Classical examples include: splitting the symbol $s$ along its longest dimension only; and uniformly splitting $s\subseteq\R^n$ into $2^n$ subsymbols ($2$ per dimension).

\begin{lemma}[\cite{meyer2017nahs}]
\label{lemma controller}
If Algorithm~\ref{algo global} terminates for a plan $\Lambda^i$, then the controller $C^i:X_\tau\rightarrow U_\tau$ such that $C^i(z)=C_a^i(s)$ for all $z\in s$ solves Problem~\ref{pb refinement}, i.e.\
the closed loop of $S_\tau$ with transitions $z'\in\delta_\tau(z,C^i(z))$ follows $\Lambda^i$ when starting in $\sigma^i_0$.
\end{lemma}

The proof of Lemma~\ref{lemma controller} in~\cite{meyer2017nahs} is independent of the over-approximation method associated with system (\ref{eq system}).
A solution to the main LTL control problem then immediately follows.
\begin{coro}
\label{coro controller}
If Algorithm~\ref{algo global} terminates for all plans $\Lambda^i$ derived in Section~\ref{subsub prelim plan}, then the controller $C:\N\times X_\tau\rightarrow U_\tau$ defined by $C(i,z)=C^i(z)$ solves Problem~\ref{pb LTL}.
\end{coro}
% Although the hierarchical decomposition in Section~\ref{sub prelim hierarchical} enables the consideration of general high-level control problems (both in terms of specification and dynamics) with a reduced complexity, guarantees for the converse implication of Corollary~\ref{coro controller} (``if Problem~\ref{pb LTL} can be solved on $S_\tau$, then the problem decomposition and Algorithm~\ref{algo global} will find a controller solving it'') cannot be provided in general due to both the hierarchical decomposition of the problem and the use of over-approximations in the definition of the abstractions $S_a^i$.

%%%%%%%%%%%%%%%%%%%%%%%%%%%%%%%%%%%%%%%%%%%%%%%%%%%%%%%%%%%%%%%%%%%%%%%%%%%%%%%%%%%%%%%%%%%%%%%%%%%%%%%%%%%%
%%%%%%%%%%%%%%%%%%%%%%%%%%%%%%%%%%%%%%%%%%%%%%%%%%%%%%%%%%%%%%%%%%%%%%%%%%%%%%%%%%%%%%%%%%%%%%%%%%%%%%%%%%%%
\section{Application to non-holonomic motion planning}
\label{sec simulation}
In this section, we consider a high-level motion planning problem for a mobile robot evolving in an office environment.
The robot is modeled by disturbed unicycle dynamics:
\begin{equation}
\label{eq unicycle}
\dot z=f(z,u,d)=\begin{pmatrix}
v\cos(\theta)+d_1\\
v\sin(\theta)+d_2\\
\omega+d_3
\end{pmatrix}
\end{equation}
where $z=(x,y,\theta)\in\mathcal{Z}\subseteq\R^3$ is the state (2D position and orientation), $u=(v,\omega)\in\mathcal{U}\subseteq\R^2$ is the control input (linear and angular velocities) and $d=(d_1,d_2,d_3)\in\mathcal{D}\subseteq\R^3$ is the disturbance.
We further assume that the disturbance take its values in an interval $\mathcal{D}=[\underline{d},\overline{d}]$ of $\R^3$.
We first apply the proposed reachability analysis results to the unicycle (\ref{eq unicycle}) in Section~\ref{sub simu reachability}.
Section~\ref{sub simu problem} then describes the considered motion planning problem and the associated simulation results.

%%%%%%%%%%%%%%%%%%%%%%%%%%%%%%%%%%%%%%%%%%%%%%%%%%%%%%
\subsection{Reachability analysis}
\label{sub simu reachability}
We first define function $g:\mathcal{Z}\times\mathcal{U}\times\mathcal{D}\times\mathcal{Z}\times\mathcal{U}\times\mathcal{D}\rightarrow\R^3$ as in (\ref{eq decomposition}) for the unicycle model (\ref{eq unicycle}).
Since all partial derivatives of $f_3$ are non-negative ($a^c_{3j}\geq0$ as in (C1) of Section~\ref{sub ar reachability} for all $c\in\{z,u,d\}$), we thus have $\alpha_3^z=\alpha_3^d=\begin{pmatrix}0 & 0 & 0\end{pmatrix}$, $\alpha_3^u=\begin{pmatrix}0 & 0\end{pmatrix}$, $Z_3=z$, $U_3=u$ and $D_3=d$, leading to:
\begin{equation}
\label{eq simu decomposition theta}
g_3(z,u,d,z^*,u^*,d^*)=f_3(z,u,d)=\omega+d_3.
\end{equation}

For abstraction $S_a^i$ in Section~\ref{sub ar refinement}, the over-approximation is computed with a known and constant control value $u$ over the time period $[0,\tau]$.
The signs of $\frac{\partial f_i}{\partial v}$ for $i\in\{1,2\}$ thus have no influence on this over-approximation since we always have $u=u^*$ in (\ref{eq decomposition}), implying $U_1=U_2=u$ and $\alpha^u_i(u-u^*)=0$.
Since for $i,j\in\{1,2\}$ and $k\in\{1,2,3\}$ the partial derivatives $\frac{\partial f_i}{\partial z_j}$ and $\frac{\partial f_i}{\partial d_k}$ are non-negative, we have $\alpha_{i1}^z=\alpha_{i2}^z=0$, $\alpha_i^d=\begin{pmatrix}0 & 0 & 0\end{pmatrix}$ and $D_i=d$, and we thus obtain for $i\in\{1,2\}$:
\begin{equation}
\label{eq simu decomposition xy}
g_i(z,u,d,z^*,u,d^*)=f_i(Z_i,u,d) + \alpha_{i3}^z(\theta-\theta^*),
\end{equation}
where $Z_1$, $Z_2$, $\alpha_{13}^z$ and $\alpha_{23}^z$ are defined as in Section~\ref{sub ar reachability} from the values of the four bounds $a^z_{13}$, $b^z_{13}$, $a^z_{23}$ and $b^z_{23}$ of the remaining two partial derivatives whose signs are not constant: 
\begin{equation}
\label{eq simu partial bounds}
\frac{\partial f_1}{\partial\theta}=-v\sin(\theta)\in[a^z_{13},b^z_{13}],
\quad\frac{\partial f_2}{\partial\theta}=v\cos(\theta)\in[a^z_{23},b^z_{23}].
\end{equation}

Considering $\theta\in(-\pi,\pi]$ would result in too conservative global bounds $\frac{\partial f_1}{\partial\theta},\frac{\partial f_2}{\partial\theta}\in[-v,v]$.
Instead, we follow Remark~\ref{rmk bounds} to find a subset of possible orientations on each sampling period $[0,\tau]$ and thus obtain tighter local bounds in (\ref{eq simu partial bounds}).
Given an interval of initial orientations $[\underline{\theta_0},\overline{\theta_0}]\subseteq(-\pi,\pi]$ and a known angular velocity $\omega$, (\ref{eq unicycle}) gives $\dot\theta\in[\omega+\underline{d}_3,\omega+\overline{d}_3]$, and thus the orientation $\theta(\tau)$ at time $\tau>0$ is bounded as $\theta(\tau)\in[\underline{\theta_0}+\tau(\omega+\underline{d}_3),\overline{\theta_0}+\tau(\omega+\overline{d}_3)]$.
\ifdouble
  Over the whole sampling period $[0,\tau]$, we obtain the following set $[\underline{\theta},\overline{\theta}]$ of possible orientations: $\theta([0,\tau])\in[\underline{\theta},\overline{\theta}]=[\underline{\theta_0}+\min(0,\tau(\omega+\underline{d}_3)),\overline{\theta_0}+\max(0,\tau(\omega+\overline{d}_3))]$.
%   Over the whole sampling period $[0,\tau]$, we obtain the following set $[\underline{\theta},\overline{\theta}]$ of possible orientations $\theta([0,\tau])\in[\underline{\theta},\overline{\theta}]$:
%   \begin{equation}
%   \label{eq simu orientation}
%   [\underline{\theta},\overline{\theta}]=[\underline{\theta_0}+\min(0,\tau(\omega+\underline{d}_3)),\overline{\theta_0}+\max(0,\tau(\omega+\overline{d}_3))].
%   \end{equation}
\else
  Over the whole sampling period $[0,\tau]$, we obtain the following set $[\underline{\theta},\overline{\theta}]$ of possible orientations:
  \begin{equation}
  \label{eq simu orientation}
  \theta([0,\tau])\in[\underline{\theta},\overline{\theta}]=[\underline{\theta_0}+\min(0,\tau(\omega+\underline{d}_3)),\overline{\theta_0}+\max(0,\tau(\omega+\overline{d}_3))].
  \end{equation}
\fi

When $v\geq0$, the bounds in (\ref{eq simu partial bounds}) can thus be computed by:
\begin{equation}
\label{eq simu bounds}
\begin{split}
a^z_{13}=-v\max_{\theta\in[\underline{\theta},\overline{\theta}]}(\sin(\theta)),\qquad & a^z_{23}=v\min_{\theta\in[\underline{\theta},\overline{\theta}]}(\cos(\theta))\\
b^z_{13}=-v\min_{\theta\in[\underline{\theta},\overline{\theta}]}(\sin(\theta)),\qquad & b^z_{23}=v\max_{\theta\in[\underline{\theta},\overline{\theta}]}(\cos(\theta))
\end{split}
\end{equation}
with swapped $\min$ and $\max$ operators when $v<0$.
Since $[\underline{\theta},\overline{\theta}]\cap(-\pi,\pi]\neq\emptyset$, the extrema of the $\cos$ and $\sin$ are:
\begin{equation}
\min_{\theta\in[\underline{\theta},\overline{\theta}]}(\cos(\theta))=
\begin{cases}
-1 \hfill\text{if }\{-\pi,\pi\}\cap[\underline{\theta},\overline{\theta}]\neq\emptyset\\
\min(\cos(\underline\theta),\cos(\overline\theta)) \hfill\quad\text{otherwise}
\end{cases}\label{eq max sin}
\end{equation}
with similar equations replacing $\{-\pi,\pi\}$ by $\{-2\pi,0,2\pi\}$, $\{\frac{-5\pi}{2},\frac{-\pi}{2},\frac{3\pi}{2}\}$ and $\{\frac{-3\pi}{2},\frac{\pi}{2},\frac{5\pi}{2}\}$ for $\max(\cos(\theta))=1$, $\min(\sin(\theta))=-1$ and $\max(\sin(\theta))=1$, respectively.
% \begin{align}
% \min_{\theta\in[\underline{\theta},\overline{\theta}]}(\cos(\theta))&=
% \begin{cases}
% -1 \hfill\text{if }\{-\pi,\pi\}\cap[\underline{\theta},\overline{\theta}]\neq\emptyset\\
% \min(\cos(\underline\theta),\cos(\overline\theta)) \hfill\quad\text{otherwise}
% \end{cases}\nonumber\\
% %
% \max_{\theta\in[\underline{\theta},\overline{\theta}]}(\cos(\theta))&=
% \begin{cases}
% 1 \hfill\text{if }\{-2\pi,0,2\pi\}\cap[\underline{\theta},\overline{\theta}]\neq\emptyset\\
% \max(\cos(\underline\theta),\cos(\overline\theta)) \hfill\quad\text{otherwise}
% \end{cases}\nonumber\\
% %
% \min_{\theta\in[\underline{\theta},\overline{\theta}]}(\sin(\theta))&=
% \begin{cases}
% -1 \ \hfill\text{if }\{\frac{-5\pi}{2},\frac{-\pi}{2},\frac{3\pi}{2}\}\cap[\underline{\theta},\overline{\theta}]\neq\emptyset\\
% \min(\sin(\underline\theta),\sin(\overline\theta)) \hfill\quad\text{otherwise}
% \end{cases}\nonumber\\
% %
% \max_{\theta\in[\underline{\theta},\overline{\theta}]}(\sin(\theta))&=
% \begin{cases}
% 1 \hfill\text{if }\{\frac{-3\pi}{2},\frac{\pi}{2},\frac{5\pi}{2}\}\cap[\underline{\theta},\overline{\theta}]\neq\emptyset\\
% \max(\sin(\underline\theta),\sin(\overline\theta)) \hfill\quad\text{otherwise.}\label{eq max sin}
% \end{cases}
% \end{align}

The bounds $a^z_{13}$, $b^z_{13}$, $a^z_{23}$ and $b^z_{23}$ computed from (\ref{eq simu bounds})-(\ref{eq max sin}) for each sampling period lead to function $g$ defined in (\ref{eq simu decomposition theta})-(\ref{eq simu decomposition xy}) followed by the duplicated dynamical system (\ref{eq system duplicated}) with vector field $h$ and trajectories $\Phi_h$ as in Section~\ref{sub ar reachability}.
Theorem~\ref{th over approximation} with state interval $[\underline{s},\overline{s}]\subseteq\mathcal{Z}$ and constant control $u\in\mathcal{U}$ gives the over-approximation $[\Phi_h^1(\tau,\underline{z},u,\underline{d},\overline{z},u,\overline{d}),\Phi_h^2(\tau,\underline{z},u,\underline{d},\overline{z},u,\overline{d})]$ as used for the abstraction refinement in Section~\ref{sub ar refinement}.

%%%%%%%%%%%%%%%%%%%%%%%%%%%%%%%%%%%%%%%%%%%%%%%%%%%%%%
\subsection{Problem description and simulation results}
\label{sub simu problem}
We consider a high-level motion planning problem for a mobile robot evolving in a $33\times 20$ square meters office environment.
This 2D workspace is formed by four rooms and a central hallway, as sketched in Figure~\ref{fig plan4} uniformly partitioned into $20\times 12$ cells and where the black cells represent static obstacles (walls).
The four regions of interest (in blue) denoted as $\pi_1$ to $\pi_4$ correspond to the cells in which the observation tasks of each room are to be carried out.

The robot is modeled as a unicycle (\ref{eq unicycle}) where the state $z=(x,y,\theta)$ evolves in $\mathcal{Z}=[0,33]\times[0,20]\times(-\pi,\pi]$, the control inputs $u=(v,\omega)$ are picked in the discrete set $U_a=\{-0.5,-0.25,0,0.25,0.5\}\times\{-0.3,-0.15,0,0.15,0.3\}$ and the disturbances lie in $\mathcal{D}=[-0.05,0.05]\times[-0.05,0.05]\times[-0.03,0.03]$.
The initial state of (\ref{eq unicycle}) is taken in the cell $\pi_1$.

The control objective is expressed by the LTL formula $\varphi=\square\lozenge \pi_2 \wedge \square\lozenge \pi_4 \wedge \lozenge \pi_3 \wedge \neg \pi_3\mathcal{U} \pi_4$, whose first two elements are a surveillance task in $\pi_2$ and $\pi_4$ (visit each infinitely often) and the last two elements mean that we also want to eventually visit $\pi_3$ but not before $\pi_4$ has been visited at least once.
The obstacle avoidance is not included in $\varphi$ since it is automatically handled in the second layer of the hierarchical decomposition.

The simulation results are obtained on a laptop with a $1.7$ GHz CPU and $4$ GB of RAM (on Matlab for steps $2$ and $3$).
% The simulation results in this section are obtained on a laptop with a $1.7$ GHz CPU and $4$ GB of RAM.
% The computations for steps $2$ and $3$ are done on Matlab.

\subsubsection{First layer - LTL planning}
\label{subsub simu layer1}
As in Section~\ref{subsub prelim LTL}, we first define the finite transition system $S_{\Pi}=(\Pi,\delta_{\Pi})$, where $\Pi=\{\pi_1,\pi_2,\pi_3,\pi_4\}$ is the set of regions of interest and $\delta_{\Pi}\subseteq\Pi^2$ represents the office structure in Figure~\ref{fig plan4} such that room $1$ can only be reached through room $2$: $\delta_{\Pi}=\Pi^2\backslash\{(\pi_1,\pi_3),(\pi_3,\pi_1),(\pi_1,\pi_4),(\pi_4,\pi_1)\}$.
The LTL planning is solved in $2$ milliseconds (due to the small size of $S_\Pi$) by the model checker P-MAS-TG~\cite{guo2015multi} resulting in an infinite accepting path $\bar\pi=\pi_1\pi_2\pi_4\pi_3(\pi_2\pi_4)^\omega$ where the prefix $\pi_1\pi_2\pi_4\pi_3$ is followed once and the suffix $\pi_2\pi_4$ is repeated infinitely often.
This infinite path can then be handled by applying the next two layers to only five pairs of regions of interest: $\pi_1-\pi_2$, $\pi_2-\pi_4$, $\pi_4-\pi_3$, $\pi_3-\pi_2$ and $\pi_4-\pi_2$.
% Although the pair $\pi_2-\pi_4$ appears in both the prefix and the suffix of $\bar\pi$, the next two steps only need to be applied once for this pair.

\subsubsection{Second layer - Physical environment}
\label{subsub simu layer2}
The transition system $S_{\mathcal{P}}=(\mathcal{P},\delta_{\mathcal{P}})$ is then defined as in Section~\ref{subsub prelim plan} for the planning and obstacle avoidance in the 2D workspace (partitioned in $20\times 12$ \emph{cells}) while disregarding the system dynamics (\ref{eq unicycle}).
For each of the above $5$ pairs $\pi_i-\pi_j$, one of the shortest discrete plan $\Lambda^{ij}=\sigma^{ij}_0\sigma^{ij}_1\dots\sigma^{ij}_{r_{ij}}$ in $\mathcal{P}$ connecting $\sigma^{ij}_0=\pi_i$ and $\sigma^{ij}_{r_{ij}}=\pi_j$ is obtained by applying a Breadth-First Search algorithm~\cite{cormen2009introduction} on $S_{\mathcal{P}}$, with an average computation time of $46$ milliseconds per pair $\pi_i-\pi_j$.

\subsubsection{Third layer - Control synthesis}
\label{subsub simu layer3}
Taking inspiration from the guidelines in~\cite{meyer2017ifac}, we pick the sampling period $\tau = 1.2*\max(33/20,20/12)/0.5=4$ seconds approximating the minimal time to translate vertically or horizontally a cell in $\mathcal{P}$ to one of its neighbors using maximal linear velocity $u=(0.5,0)$.
The factor $1.2$ multiplies this value to account for (\ref{eq unicycle}) not always keeping axis-aligned orientations $\theta\in\{-\pi/2,0,\pi/2,\pi\}$ or using the maximal linear velocity.

Since the control objective $\Lambda^{ij}$ for this third layer is in the 2D workspace while system (\ref{eq unicycle}) has a 3D state, we consider a modified definition for the valid set of a 2D cell $\sigma\in\mathcal{P}$:
\ifdouble
  $V_{2D}^i(\sigma)=\{s\subseteq\sigma~|~\exists s_\theta\subseteq(-\pi,\pi], s\times s_\theta\in X_a^i\cap V^i(\sigma\times(-\pi,\pi])\}$,
\else
  \begin{equation}
  \label{eq valid unicycle}
  V_{2D}^i(\sigma)=\{s\subseteq\sigma~|~\exists s_\theta\subseteq(-\pi,\pi], s\times s_\theta\in X_a^i\cap V^i(\sigma\times(-\pi,\pi])\},
  \end{equation}
\fi
i.e.\ $V_{2D}^i(\sigma)$ contains any 2D projection $s$ of a 3D symbol $s\times s_\theta$ belonging to both the refined partition $X_a^i$ and the valid set of the 3D cell $\sigma\times(-\pi,\pi]$ as in Section~\ref{sub ar refinement}.
As a result, the abstraction refinement still works on the partition $X_a^i$ of the 3D state space but becomes more permissive since the 2D control synthesis (using $V_{2D}^i$) disregards the validity of the orientation.
The drawback is that after each application of the controller, we may need to rotate the robot to reach a valid orientation (which always exists by definition of $V_{2D}^i$).

Function \texttt{Pick} in Algorithm~\ref{algo global} is chosen similarly to~\cite{meyer2017acc} as a queue picking the oldest cell of $\Lambda^{ij}$ added to the queue among the least refined ones.
The orientation interval $(-\pi,\pi]$ is initially partitioned into $4$ identical intervals.
Function \texttt{Split}$(s)$ in Algorithm~\ref{algo global} then takes a uniform partition of the 3D symbol $s$ into $8$ subsymbols ($2$ per dimension).

For the $5$ plans $\Lambda^{12}$, $\Lambda^{24}$, $\Lambda^{43}$, $\Lambda^{32}$ and $\Lambda^{42}$ from layer $2$, the control synthesis takes between $18$ and $58$ minutes ($39$ minutes per plan on average) with a number of refinement iterations ranging from $40$ to $76$ ($54$ per plan on average).

\subsubsection{Simulation results}
\label{subsub simu results}
Figure~\ref{fig plan4} provides a visualization of the abstraction refinement results for plan $\Lambda^{32}$, where the finer black grid is the 2D projection of the refined partition $X_a^i$ and the red area is the 2D valid sets $V_{2D}^i(\sigma)$ (not represented on cells $\pi_2$ and $\pi_3$).
In this particular case, the valid sets happen to cover the whole cells: $V^i(\sigma^i_k)=P_a^i(\sigma^i_k)$ for all $\sigma^i_k$.

The disturbed unicycle (\ref{eq unicycle}) in closed-loop with the global controller $C$ from Corollary~\ref{coro controller} is then simulated from an initial state $z_0$ randomly picked in the 3D cell $\pi_1\times(-\pi,\pi]$.
At each time step, the robot measures its position $z=(x,y,\theta)$ and finds the corresponding 3D symbol $s_{3D}=s_{2D}\times s_\theta\in X_a^i$ such that $z\in s_{3D}$.
Since the control synthesis was successful for the current 2D cell $\sigma_{2D}\in\mathcal{P}$ with $s_{2D}\subseteq\sigma_{2D}$, then by definition of $s_{2D}\in V_{2D}^i(\sigma_{2D})$ there exists a valid 3D symbol $s_{3D}'=s_{2D}\times s_\theta'\in X_a^i\cap V^i(\sigma_{2D}\times(-\pi,\pi])$.
If $s_{3D}\neq s_{3D}'$, we apply a constant rotation $u=(0,\omega)$ until the system reaches a new state $z'=(x,y,\theta')\in s_{3D}'$ (assuming that we have $d_1=d_2=0$ during such rotations only).
Whether a rotation was done or not, we finally apply the constant control value $C_a^i(s_{2D})$ for $\tau=4$ seconds to go to the next cell of the current plan and repeat this procedure with a new measurement of the state.
The closed-loop trajectory for plan $\Lambda^{32}$ is displayed in green in Figure~\ref{fig plan4}.
The snaps in this trajectory correspond to rotations before applying the next control, while smoother sections over several cells mean that no rotation was needed.

\begin{figure}[htb]
    \centering
    \includegraphics[width=\columnwidth]{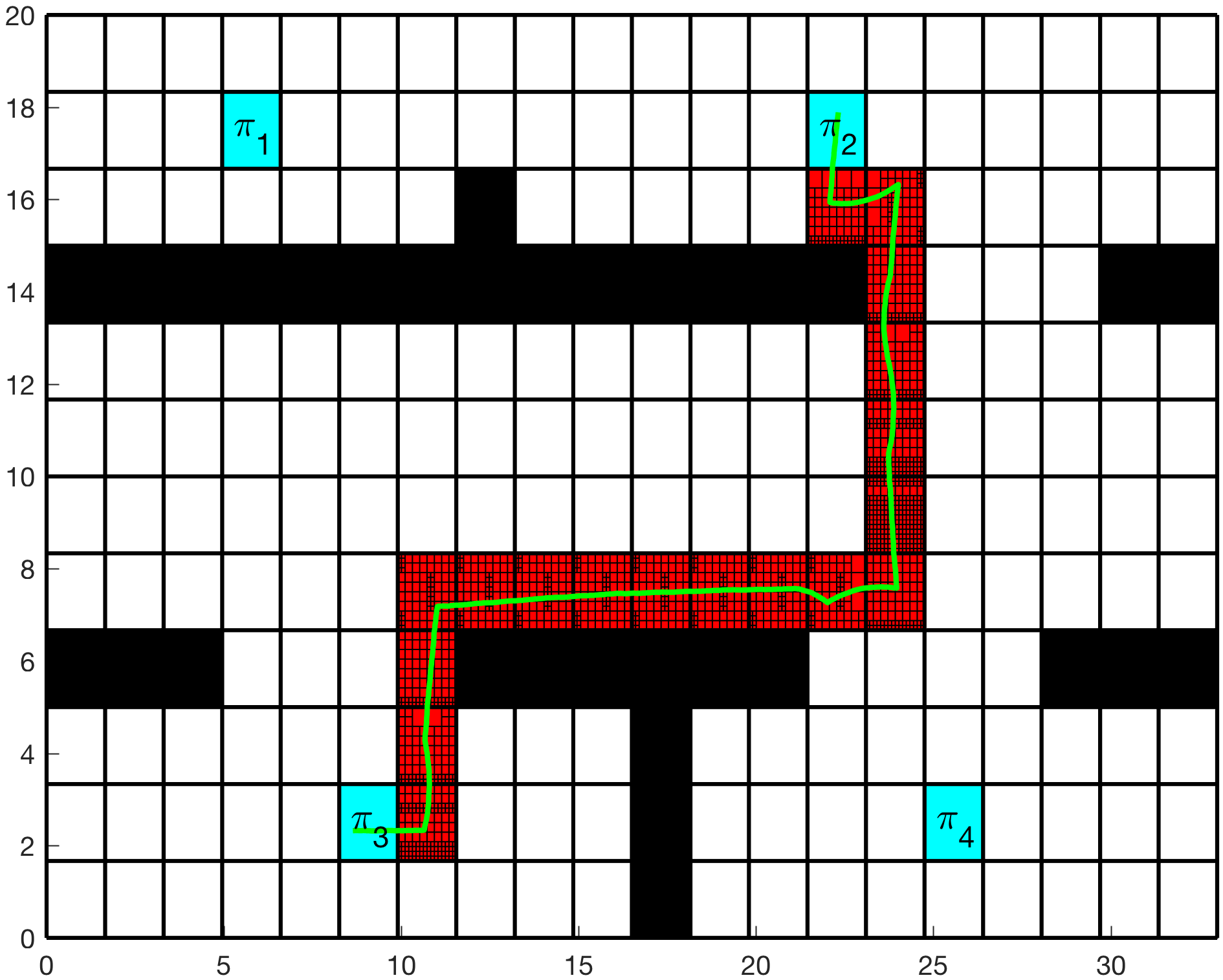}
    \caption{Partitioned office environment with obstacles (black) and four regions of interest (blue) as the center of each room. For plan $\Lambda^{32}$ from $\pi_3$ to $\pi_2$, we also display the refined partition (finer black grid), the valid symbols (red) and the closed-loop trajectory (green).}
    \label{fig plan4}
\end{figure}

\subsubsection{Final comments}
\label{subsub simu comments}
As mentioned in Section~\ref{sub prelim hierarchical}, the solver of each layer can be substituted by any existing tool designed to tackle the same subproblem.
For the example of this section, we could also consider replacing layers $1$ and $2$ by conPAS2~\cite{conpas2} for a fully actuated system $\dot z=u$, or layers $2$ and $3$ by the unicycle example provided in PESSOA~\cite{pessoa}.
On the other hand, none of the existing tools mentioned in Section~\ref{sec intro} are capable of solving the whole control problem tackled in this section, because they do not handle either nonlinear systems~\cite{tulip,conpas2,ltlmop} or general LTL specifications~\cite{pessoa,cosyma,scots}.
As a consequence, the $3$-layer hierarchical framework introduced in Section~\ref{sub prelim hierarchical} is strictly more general than these tools and thus cannot be compared with them on the case study of this section for the disturbed unicycle (\ref{eq unicycle}) under the LTL specification $\varphi$.

To highlight the novelty of the proposed approach, we can however discuss the complexity for approaching the same problem with abstraction-based methods without relying on the hierarchical decomposition or abstraction refinement.
With no synthesis of an accepting path $\bar\pi$ for the LTL formula and of the discrete plans $\Lambda^{ij}$ realizing this path (as done in layers $1$ and $2$), a symbolic abstraction needs to be created for the whole state space.
The creation of such abstraction takes over $43$ hours (compared to $3$ hours in our framework) when using the finest partition granularity reached in the abstraction refinement layer to ensure that the controller obtained above could be reproduced in this setting.
Further intensive computation would also result from attempting a controller synthesis with respect to the LTL specification on this abstraction containing over $98$ millions state-input pairs.
This last step could not even be attempted as the abstraction variable weighted over 4GB which could not be stored in Matlab.

%%%%%%%%%%%%%%%%%%%%%%%%%%%%%%%%%%%%%%%%%%%%%%%%%%%%%%%%%%%%%%%%%%%%%%%%%%%%%%%%%%%%%%%%%%%%%%%%%%%%%%%%%%%%
%%%%%%%%%%%%%%%%%%%%%%%%%%%%%%%%%%%%%%%%%%%%%%%%%%%%%%%%%%%%%%%%%%%%%%%%%%%%%%%%%%%%%%%%%%%%%%%%%%%%%%%%%%%%
\section{Conclusion}
\label{sec conclu}
The first contribution of this paper is a three-layer hierarchical decomposition of a high-level control problem under a Linear Temporal Logic formula by iteratively solving finer versions of the problem: first solve the problem only on the regions of interest involved in the LTL formula, then realize the obtained sequence of regions by finding discrete plans in the partitioned workspace, finally synthesize a controller for the dynamical system to follow these plans.
This framework enables the consideration of general LTL control problems for nonlinear systems and the subproblems defined at each layer can be solved through various existing tools.
The second contribution is a new method to over-approximate the finite-time reachable set of any continuously differentiable system, relying on an  auxiliary monotone system obtained by using Jacobian bounds to compensate the non-monotone components of the initial system.
% The generality of this result naturally comes at the cost of an increased conservativeness.
For generality of the hierarchical framework to continuously differentiable systems, an implementation for layer 3 is proposed using this new reachability analysis result within an abstraction refinement algorithm.

% Since the proposed three-layer approach efficiently provides a solution to the first 2 steps, the main goal for future work is to improve the efficiency of the third step using abstraction refinement which is the main bottleneck.
% In particular, for dynamical systems of large dimensions, we will consider a compositional abstraction refinement approach similar to~\cite{meyer2017nahs} where the dynamical system is first decomposed into subsystems (each considering a subset of state and control variables) before applying the abstraction method to each subsystem at a lower computational cost.
% For a greater re-usability of this solution minimizing inputs from the user, this decomposition of the dynamics should ideally be automated.

The main goal for future work is to strengthen the cohesion in the implementation of the three control layers to distribute this framework as a coherent and publicly available tool.

%%%%%%%%%%%%%%%%%%%%%%%%%%%%%%%%%%%%%%%%%%%%%%%%%%%%%%%%%%%%%%%%%%%%%%%%%%%%%%%%%%%%%%%%%%%%%%%%%%%%%%%%%%%%
%%%%%%%%%%%%%%%%%%%%%%%%%%%%%%%%%%%%%%%%%%%%%%%%%%%%%%%%%%%%%%%%%%%%%%%%%%%%%%%%%%%%%%%%%%%%%%%%%%%%%%%%%%%%
\bibliographystyle{abbrv}
\bibliography{Literature} 

\end{document}